\documentclass{article}
\usepackage[a4paper, total={6in, 8in}]{geometry}
\usepackage{graphicx} 
\usepackage{amsmath}
\usepackage{amssymb}
\usepackage{natbib}
\usepackage{float}
\usepackage{algorithm}
\usepackage{algpseudocode}
\usepackage{xcolor}
\usepackage{booktabs}
\usepackage{multirow}
\usepackage{url}

\usepackage{comment}
\usepackage{cleveref}
\usepackage{amsthm}

\newcommand{\ind}{\perp\!\!\!\!\perp}

\newcommand{\E}[1]{\mathbb{E}\left[ #1 \right]}

\newtheorem{theorem}{Theorem}
\newtheorem{assumption}{Assumption}
\newtheorem{remark}{Remark}
\newtheorem{lemma}{Lemma}

\crefname{example}{Example}{Examples}
\crefname{definition}{Definition}{Definitions}
\crefname{proposition}{Proposition}{Propositions}
\crefname{theorem}{Theorem}{Theorems}
\crefname{assumption}{Assumption}{Assumptions}
\crefname{remark}{Remark}{Remarks}
\crefname{lemma}{Lemma}{Lemmas}

\usepackage{authblk}

\title{Inference on counterfactual distributions using martingale posteriors}
\author{Gregor Steiner and Mark Steel}
\affil{Department of Statistics, University of Warwick}

\begin{document}

\maketitle

\begin{abstract}
Causal inference is often focused on average effects, which can hide important aspects of the effect distributions. Here we consider the entire posterior effects distribution by estimating full counterfactual outcome distributions. We propose a methodology for inference on counterfactual distributions which builds upon the martingale posterior framework of \cite{fong_martingale_2023}. This provides a highly flexible approach to estimating densities, distribution functions, and derived quantities such as quantiles, which coherently quantifies the  epistemic uncertainty on any target estimand of interest. As the predictive recursions are based on an underlying nonparametric model (a Dirichlet process mixture model), our method naturally inherits robustness with respect to restrictive parametric assumptions. In addition, implementation of our method is typically very fast. This approach can be applied to marginal or conditional counterfactual distributions and is easily extended to an instrumental variables setup. Using the concept of almost conditionally identically distributed random variables, we prove convergence of the martingale posterior  inference on the counterfactual outcome distributions for the causal models considered in the paper.  We illustrate our approach on both simulated and real data. Using the latter, we investigate  the effect of zinc lozenges on common cold duration, the impact of vitamin A supplementation on children's survival rates with one-sided non-compliance \citep[analysed in][]{imbens_bayesian_1997} and the effect of job training \citep{lalonde_evaluating_1986}.
\end{abstract}

\section{Introduction}

Causal inference traditionally focuses on mean effects. A popular estimand is the average treatment effect, which quantifies the difference between mean outcomes of a treated and a control population. While this can be a useful summary, relying on mean effects alone may miss important consequences of the treatment. A treatment might benefit some individuals while harming others to the same degree, producing a bimodal outcome distribution with an unchanged mean. It might instead have a ``lottery'' effect, raising the probability of extreme positive outliers while lowering outcomes for the majority. We illustrate two such scenarios in \Cref{fig:motivating_example}. In both, mean-based inference would find no effect, whereas estimating counterfactual distributions reveals the true underlying situation.

\begin{figure}
    \centering
    \includegraphics[width=0.9\linewidth]{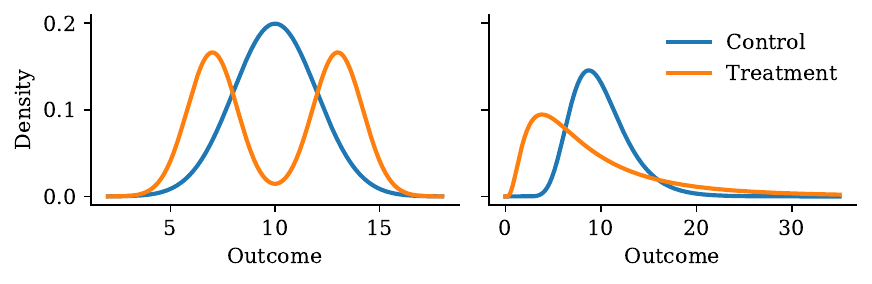}
    \caption{The left panel illustrates a polarisation scenario where a unimodal control distribution shifts to a bimodal treatment distribution. The right panel illustrates a lottery scenario, where treatment increases skewness. In both cases, the mean remains identical, demonstrating how mean-based estimators can mask significant treatment-induced heterogeneity.}
    \label{fig:motivating_example}
\end{figure}

Much of the earlier literature on distributional causal inference studies counterfactual outcomes through cumulative distribution functions or quantiles. \cite{chernozhukov2013inference} develop estimators and inference for counterfactual distributions using quantile regression or related regression-based decompositions, and this line was extended to formal inference for distribution and quantile functions in treatment effect models \citep{donald2014estimation} and to continuous-treatment settings \citep{ai2022estimation}.

More recent work targets the counterfactual density directly. On the frequentist side, this includes semiparametric and doubly robust estimators \citep{kennedy2023semiparametric}, shape-constrained inference under log-concavity \citep{ham2024doublyrobustestimationinference}, kernel-Stein based estimators \citep{martinez2024counterfactual}, and deep generative approaches such as interventional normalizing flows \citep{melnychuk2023normalizing}. These frequentist and machine-learning methods are often geared toward efficient point estimation and asymptotic inference, while epistemic uncertainty quantification is less central or less fully developed, with \cite{ham2024doublyrobustestimationinference} being a notable exception. The Bayesian literature, by contrast, more naturally propagates uncertainty by placing flexible models on the joint or conditional outcome distribution and then treating densities, quantiles, and other distributional features as posterior functionals \citep[e.g.][]{roy2018bayesian, xu2022bayesiansemiparametricmethodestimating}. A recurring feature across these strands is that a given method typically targets one distributional representation and one estimand at a time.

In instrumental variable (IV) settings with unobserved confounding, previous work has also focused mainly on quantiles and distribution functions \citep[e.g.][]{chernozhukov_iv_2005, kook_instrumental_2025}. \cite{jung2021} construct double machine learning estimators for complier interventional distributions. More recently, \cite{holovchak_distributional_2025} study the entire interventional distribution in a general IV setting by first learning the conditional distribution through energy-score minimisation and then drawing samples from the structural model.

We propose a method for inference on counterfactual distributions grounded in the martingale posterior framework \citep{fong_martingale_2023}. Epistemic uncertainty on densities, distribution functions, or any other functional can be coherently quantified through the martingale posterior samples. We rely on flexible, nonparametric prediction rules, so the method avoids restrictive parametric assumptions and inherits a corresponding robustness. Our method works under standard unconfoundedness assumptions and extends to instrumental variable designs with unobserved confounding. In both cases, marginal and conditional counterfactual distributions can be targeted. On the theoretical side, we establish convergence of the predictive recursions. To the best of our knowledge, this is the first application of martingale posteriors to causal counterfactual distributions. \cite{melnychuk_frequentist_2026} also use martingale posteriors in a causal context, but to analyse the frequentist consistency of prior-data fitted networks in order to calibrate uncertainty for the average treatment effect, whereas we consider full counterfactual distributions.

Our approach performs well on simulated data. We apply our method to real data in order to investigate the effect of zinc lozenges on common cold duration, the impact of vitamin A supplementation on children's survival rates with one-sided non-compliance \citep[analysed in][]{imbens_bayesian_1997} and the effect of job training \citep{lalonde_evaluating_1986}. The implementation of our method is typically very fast, and the code is freely available at \url{https://github.com/gregorsteiner/CausalMP}.

\section{Preliminaries}

\subsection{Setting and notation}

Let $(Y, X, W) \sim P$ be random variables drawn from a joint distribution $P$, where $Y$ is the outcome, $X$ is the treatment (or endogenous variable), and $W$ is a vector of (exogenous) control variables. Throughout, we assume $P$ admits a density\footnote{This will denote probability density functions for dimensions with continuous distributions and  probability mass functions for those with discrete distributions.} and use the generic symbol $p$ for all density functions, where the arguments specify the variables of interest. For instance, the joint density can be written as $p(y, x, w) = p(y \mid x, w) p(x, w)$. Throughout, we use capital letters ($Y, X, W$, and later $Z$) for random variables and lower-case letters ($y, x, w, z$) for their realisations, whether these are observed data, generic density arguments, or values drawn within a specific predictive sequence. Let $\mathcal{D}_{1:n} = \{(y_i, x_i, w_i) \}_{i=1}^n$ be the observed dataset of size $n$, where we sometimes drop the subscripts for convenience; we reserve $\mathcal{D}$ (with a range subscript) to denote such a collection of realised tuples. The main interest is in a parameter or estimand $\theta \in \Theta$ that is a (potentially infinite-dimensional) functional of the distribution $P$ and we sometimes write $\theta(P)$ to emphasise this dependence. For example, we may target the interventional distribution which (under certain causal assumptions) can be expressed as
\begin{align*}
    \theta \equiv p(y(x)) = \int p(y \mid x, w) p(w) \mathrm{d}w.
\end{align*}
We take a Bayesian (or quasi-Bayesian) perspective and consider the observed data $\mathcal{D}_{1:n}$ as fixed. The aim is to perform posterior inference on $\theta$ conditional on $\mathcal{D}_{1:n}$.

\subsection{Martingale posterior distributions}

We adopt a predictive resampling perspective \citep{fortini_quasi-bayes_2020, fong_martingale_2023}, presented here for a generic (possibly vector-valued) data point $y$. The primary source of uncertainty is the failure to observe the missing observations $y_{n+1:\infty}$. Under this framework, addressing uncertainty involves imputing these unobserved data points by modelling their joint predictive distribution
$$
p(y_{n+1:\infty} \mid y_{1:n}) = \prod_{i=n+1}^\infty p(y_i \mid y_{1:i-1}).
$$
In practice, we truncate the sequence at a large but finite population size $N \gg n$, which may correspond to a known population size or reflect computational limitations. Rather than specifying an explicit likelihood and prior, we instead directly define one-step predictive updates $$p_i(y_{i+1}) = p(y_{i+1} \mid y_{1:i}),$$which are used to impute the missing observations. These predictive sequences are specified in a recursive manner through an update rule  $\phi_i$ such that for all $i=0, 1, \ldots$
\begin{align*}
    p_{i+1}(y) = \phi_i\left(p_i(y), y_{i+1} \right).
\end{align*}

The quantity of interest $\theta = \theta(P)$ is a functional of the population distribution $P$. The usual Bayesian approach would be to put a prior on $P$, either nonparametric or via a parametric family, and update it to its posterior distribution. Then, drawing from this posterior and computing $\theta(P)$ yields a posterior draw. The martingale posterior bypasses the need for a prior and likelihood by constructing a predictive sequence whose empirical distribution $P_N$ approximates the population distribution $P$. A posterior draw of $\theta$ is then obtained by computing $\theta(P_N)$ based on the generated empirical distribution. By repeating this process and generating multiple empirical distributions, we obtain a martingale posterior for the parameter of interest. The martingale posterior is equivalent to the standard Bayesian posterior when using the corresponding posterior predictive as the predictive rule to impute the missing observations. 

The martingale posterior is well-defined only if the sequence of predictive densities converges to a limiting law $P_\infty$. A sufficient condition for this is that the sequence forms a martingale \citep{fong_martingale_2023}, that is, for all $i = n+1, n+2, \ldots$
\begin{align*}
    \mathbb{E}\left[ p_i(y) \mid y_{1:i-1} \right] = \int p_i(y) \, p_{i-1}(y_i) \, \mathrm{d}y_i = p_{i-1}(y).
\end{align*}
This martingale condition on the predictive densities is equivalent to the statement that the imputed sequence $y_{n+1}, y_{n+2}, \ldots$ is \emph{conditionally identically distributed} (c.i.d.) given the observed data \citep{berti2004limit}. A c.i.d. sequence is \emph{asymptotically exchangeable} and its predictive distributions converge almost surely to a random limiting measure $P_\infty$. The condition is sufficient but not necessary, and for certain predictive rules, such as deep neural networks, verifying it theoretically is currently out of reach. Often it is enough for the sequence to be only \emph{almost} c.i.d. (a.c.i.d.), that is, to satisfy the martingale identity up to summable errors: \citet{battiston2025bayesianpredictiveinferencemartingales} show that such sequences remain asymptotically exchangeable and still admit a well-defined limiting measure. Predictive rules that do not satisfy a martingale condition may thus still induce a valid posterior, and empirical convergence diagnostics can be used to assess this in practice \citep{ng_tabmgp_2026}. In particular, monitoring the distance between $\theta(P_n)$ and the forward-sampled estimate $\theta(P_N)$ as $N$ increases is informative: stabilisation of this quantity at a non-zero constant suggests  empirically that $\theta(P_\infty)$ is well-defined, while systematic drift or divergence is a warning sign. We discuss concrete predictive rules and their properties with respect to this condition below. In addition, we formally show convergence of the martingale posterior for the causal models we introduce in the sequel. 

\subsection{Predictive rules} \label{sec:predictive_rules}

\subsubsection{The Bayesian bootstrap}

After observing data $y_{1:n}$, the Bayesian bootstrap \citep{rubin1981bayesian} characterises a random distribution through the cdf $P(y) = \sum_{i=1}^n \omega_i 1\{ y \leq y_i \}$ with weights $\omega_i \sim \mathrm{Dirichlet}(1, \ldots, 1)$. An equivalent P\'olya urn interpretation will be more useful for our purposes. Defining $P_n(y) = n^{-1} \sum_{i=1}^n 1\{ y \leq y_i \}$ as the empirical cdf, we can write the recursion for $i \geq n$
$$
P_{i+1}(y) = \frac{i}{i+1} P_i(y)  + \frac{1}{i +1 } 1\{ y \leq y_{i+1} \}.
$$
The predictive distribution is updated by drawing a point from the empirical distribution and ``reinforcing the urn'' by treating it as a new observation. For $i \to \infty$, these empirical proportions converge to the Dirichlet weights. Thus, recursively resampling with replacement treating each resampled point as a new observation yields a posterior draw from the Bayesian bootstrap. Because the resulting predictive distribution has atomic support, it is generally inappropriate as a predictive for a variable whose own distribution is continuous, since we typically want that predictive to respect the true support. It remains a convenient choice, however, when the variable being updated is not itself the target but is instead marginalised out to obtain some other functional of interest. Computing such a functional averages across the atomic support of the bootstrap predictive and smooths it out, so that it matters little whether the marginalised variable is discrete or continuous. The same reasoning does not apply when the variable of interest is continuous and its predictive distribution is the target itself, which is why we turn to the copula-based update below.

\subsubsection{Recursive copula updates}

As a continuous extension of the Bayesian bootstrap, we consider the copula update inspired by the nonparametric Dirichlet process mixture model \citep{EscobarWest95} as  proposed in \cite{hahn_recursive_2018} and further developed in \cite{fong_martingale_2023}. For a scalar outcome $y \in \mathbb{R}$ this recursive update of the predictive density $p_i(y)$ and corresponding cdf $P_i(y)$  is given by
\begin{align} \label{eq:copula_update}
\begin{aligned}
    p_{i+1}(y) &= (1- \alpha_{i+1}) p_i(y) + \alpha_{i+1} c_\rho (P_i(y), P_i(y_{i+1})) p_i(y) \\
    P_{i+1}(y) &= (1- \alpha_{i+1}) P_i(y) + \alpha_{i+1} H_\rho (P_i(y), P_i(y_{i+1})),
\end{aligned}
\end{align}
where $c_\rho$ is the bivariate Gaussian copula density and $H_\rho$ is the conditional Gaussian copula. The correlation parameter $\rho \in (0, 1)$ takes on the role of a bandwidth. As in \cite{fong_martingale_2023}, we choose $\rho$ by maximising the prequential log-score $\sum_{i=1}^n \log p_{i-1}(y_i)$. The choice of the weights $\alpha_i$ is crucial for reliable uncertainty quantification. Again, we follow \cite{fong_martingale_2023} and set $\alpha_i = (2-1/i)/(i+1)$ (see their online Appendix E.1.1 for details).

In practice, we choose an initial distribution $p_0$, typically a standard Gaussian. Then, we implement the forward steps on the observed data points $y_{1:n}$ to obtain the predictive distribution $p_n$. Based on $p_n$, we independently generate $B$ predictive sequences by recursively generating new observations $y_{i+1}$ for $i = n, n+1, \ldots$ and updating $p_i$ to $p_{i+1}$. There is no need to actually draw new observations as $Y_{i+1} \sim P_i$ implies $P_i(Y_{i+1}) \sim \mathrm{U}(0, 1)$. Thus, it is sufficient to draw uniform random variables $V_i \sim \mathrm{U}(0, 1)$, making the update computationally convenient.

\cite{fong_martingale_2023} extend the copula update to a regression setting with covariates $x \in \mathbb{R}^d$. The $\alpha_i$ weights in \Cref{eq:copula_update} are replaced by covariate dependent weights
\begin{align*}
    \alpha_{i}(x, x') = \frac{\alpha_i \prod_{j=1}^d c_{\rho_x}(\Phi(x^j), \Phi(x_j'))}{1 - \alpha_i + \alpha_i \prod_{j=1}^d c_{\rho_x}(\Phi(x^j), \Phi(x_j'))},
\end{align*}
where $\Phi$ denotes the standard Gaussian cdf and $\alpha_i$ are chosen as before. It is recommended to standardise the covariates for good performance. The additional bandwidth parameters $\rho_x$ can also be chosen to maximise the prequential log-score.

An important choice is the starting distribution $p_0$ that initiates the recursive updates. This needs to be specified by the analyst based on prior knowledge. A convenient option is standardising the response and adopting the standard Gaussian $p_0 = \mathrm{N}(y \mid 0, 1)$. Following \cite{fong_martingale_2023}, this is our default choice, even in the regression case. However, encoding prior knowledge through alternative choices is possible. One downside is that the copula recursion depends on the order of the observed data $y_{1:n}$. As in \cite{fong_martingale_2023}, we average the initial density fit over $M$ random permutations of $y_{1:n}$. We find that $M=10$ leads to stable results in our examples.

\subsubsection{Parametric predictive densities}

Assume the data are generated from a model $y_i \sim p_\psi,  i=1,\dots,n$ parameterised by $\psi \in \Psi$. A natural strategy starts from a plug-in estimate $\Hat{\psi}_n$ (e.g.\ the maximum likelihood estimate), sets $\psi_n = \Hat{\psi}_n$, and then alternates between drawing the next observation from the current plug-in predictive $y_i \sim p_{\psi_{i-1}}$ and updating the parameter via the natural gradient
\begin{align*}
    \psi_i = \psi_{i-1} + i^{-1}\, \mathcal{I}(\psi_{i-1})^{-1}\, s(\psi_{i-1}, y_i), \quad i = n+1, \ldots, N,
\end{align*}
where $s(\psi, y) = \nabla_\psi \log p_\psi(y)$ is the score function and $\mathcal{I}(\psi) = \E{s(\psi, y) s(\psi, y)^\intercal }$ is the Fisher information matrix. The parametric updating rule satisfies the martingale property as $$\E{s(\psi_{i-1}, y_i) \mid y_{1:i-1}} = 0$$
by construction. It is worth emphasising that the martingale property here holds for the \emph{parameter} sequence $\psi_i$, driven by the mean-zero score increments, and not for the predictive density $p_{\psi_i}$ itself, which is a non-linear function of the parameter and therefore not in general a martingale. Under some regularity conditions, the induced sequence of predictive densities is a.c.i.d., so that asymptotic exchangeability still holds and the martingale posterior remains well-defined \citep{battiston2025bayesianpredictiveinferencemartingales}. The learning rate of $i^{-1}$ is chosen to satisfy
\begin{align*}
    \sum_{i=n+1}^\infty i^{-1} = \infty, \quad \sum_{i=n+1}^\infty i^{-2} < \infty.
\end{align*}
Under these conditions, \citet{fong2026asymptotics} establish a predictive central limit theorem and a Bernstein-von Mises result for this class of parametric martingale posteriors.

The natural gradient $ \mathcal{I}(\psi)^{-1}\, s(\psi, y)$ is the efficient influence function for $\psi$ in the parametric model. This suggests extending the construction to semi- and nonparametric problems by replacing $\mathcal{I}(\psi)^{-1}\, s(\psi, y)$ with the efficient
influence function of the target functional, which may yield a predictive, martingale-posterior interpretation of influence function-based estimators in causal inference. We leave a detailed treatment to future work.

\section{Martingale posterior counterfactual density inference} \label{sec:mp_counterfactual}

In causal inference, there is often a discrepancy between the ``observational distribution'' that generated the data and the ``interventional distribution'' of interest. This mismatch makes it challenging to directly apply the standard martingale posterior. To make this precise, we adopt the potential outcomes framework and define $Y(x)$ as the potential outcome that would be observed if an individual were assigned treatment level $X=x$. In this setup, for each unit $i$ we observe the realised outcome $y_i = Y_i(x_i)$ corresponding to the assigned treatment level $x_i$, together with the treatment assignment $x_i$ and covariates $w_i$. Because we can never simultaneously observe $Y_i(x)$ for multiple values of $x$ on the same individual, we need additional assumptions to identify the interventional distributions of interest. In this section, we impose the standard assumptions of consistency, unconfoundedness and overlap.

\begin{assumption}[Consistency]
    The observed outcome is consistent with the treatment assignment in the sense that $Y = Y(x)$ if $X=x$.
\end{assumption}
\begin{assumption}[Unconfoundedness]
    The potential outcomes are independent of the treatment assignment conditional on the observed covariates, $Y(x) \ind X \mid W$ for all $x$.
\end{assumption}
\begin{assumption}[Overlap]
    For every covariate level, there is non-zero probability of receiving every treatment level, that is, $p(x \mid w) > 0$ for all values of $x$ and $w$ in the support of $X$ and $W$.
\end{assumption}
When these assumptions hold, the marginal interventional distribution is identified
by averaging the conditional outcome distribution over the covariate distribution,
\begin{align} \label{eq:identification}
    p(y(x)) = \int p(y(x) \mid w) p(w) \mathrm{d}w
            = \int p(y \mid x, w) p(w) \mathrm{d}w.
\end{align}
We approach inference on $p(y(x))$ from a predictive Bayesian perspective, targeting
the predictive interventional distributions $p_N(y(x))$ for every treatment level $x$
of interest. Our proposed algorithm proceeds in two main steps.

First, we build a predictive distribution for the observational model $p(y, x, w)$.
This is done recursively: at each iteration we draw a new observation from the
current predictive and then update the predictive in light of the new observation. A convenient choice is to factorise it into a conditional
outcome model $p(y \mid x, w)$, a propensity model $p(x \mid w)$, and a marginal
covariate distribution $p(w)$, and to assign a separate predictive update to each
factor. For example, we may model the continuous outcome $p(y \mid x, w)$ with the
conditional copula regression of \citet{fong_martingale_2023}, the propensity
$p(x \mid w)$ with a parametric logistic-regression update, and use the Bayesian
bootstrap for the covariate distribution $p(w)$. However, the construction is not tied to these particular choices, and any valid predictive update may be used for each factor.

Second, we recover the interventional distribution by marginalising the outcome
model over the covariate distribution 
\begin{align*}
    p_N(y(x)) = \int p_N(y \mid x, w) p_N(w) \mathrm{d}w,
\end{align*}
where $p_N(y \mid x, w)$ and $p_N(w)$ denote the final predictive distributions obtained through recursive updating. For flexible nonparametric choices of these predictive distributions, the integral will typically not be available in closed form, and numerical integration can be challenging when $W$ is high-dimensional. Therefore, we typically model the covariate distribution with the Bayesian bootstrap. Using the resampled covariate values $\{w_i\}_{i=1}^N$ generated during the forward recursion, we can approximate the integral by the Monte Carlo average
\begin{align*}
    p_N(y(x)) \approx \frac{1}{N} \sum_{i=1}^N p_N(y \mid x, w_i).
\end{align*}
This approximation requires evaluating the outcome predictive $p_N(y \mid x, w)$ at each resampled covariate value. Such pointwise evaluation is available for the conditional copula update, whose recursive scheme returns the conditional density at any conditioning value, and for parametric plug-in models, which provide the fitted conditional density in closed form.
Running $B$ independent predictive sequences and applying both steps to each yields
a martingale posterior sample of $B$ such interventional distributions, as summarised in
\Cref{alg:Causal_MP}.

\begin{algorithm}
\caption{Martingale Posterior Sampling for Interventional Distributions}
\label{alg:Causal_MP}
\begin{algorithmic}[1]
\Require Observed data $\mathcal{D}_{1:n} = \{(y_i, x_i, w_i)\}_{i=1}^n$, target treatment level $x^*$, number of sequences $B$, length of predictive sequences $N$
\State Initialise predictive rule $p_n$ from $\mathcal{D}_{1:n}$
\For{$b = 1, \ldots, B$}
    \For{$i = n+1,\ldots,N$}
        \State Draw $(y_i, x_i, w_i) \sim p_{i-1}$
        \State Update predictive rule $p_i \gets \text{Update}(p_{i-1}, (y_i, x_i, w_i))$
    \EndFor
    \State $p_N^{(b)}(y(x^*)) \gets \frac{1}{N} \sum_{i=1}^N p_N(y \mid X=x^*, W=w_i)$
\EndFor
\State \Return $\left\{ p_N^{(1)}(y(x^*)), \dots, p_N^{(B)}(y(x^*)) \right\}$
\end{algorithmic}
\end{algorithm}

The validity of \Cref{alg:Causal_MP} relies on the sequence of interventional predictive densities $p_N(y(x))$ converging to a well-defined limiting measure. This is not immediate: although $p_i(y \mid x, w)$ and $p_i(w)$ are each individually martingales in $i$ by construction, the integral over their product need not be. The following result shows that convergence nonetheless holds under mild regularity conditions. The proof, together with the precise setup and regularity conditions, is given in \Cref{app:proof_convergence}. 

\begin{theorem}[Convergence of the martingale posterior interventional density] \label{thm:convergence}
    Suppose the predictive updates for $p_i(y \mid x, w)$ and $p_i(w)$ satisfy the martingale property of \Cref{sec:predictive_rules}, and that the outcome density and the copula densities used in the updates are uniformly bounded (see \Cref{app:proof_convergence} for precise statements). Then, for every treatment level $x$, there exists a random probability measure $P_\infty(y(x))$ such that $P_i(y(x))$ converges weakly to $P_\infty(y(x))$ almost surely as $i \to \infty$.
\end{theorem}

This result is important as the martingale posterior is only a valid representation of Bayesian uncertainty if the predictive resampling scheme converges to a well-defined limiting quantity. Without it, the draws $p_N^{(b)}(y(x^*))$ produced by the algorithm would have no guarantee of stabilising as $N$ grows, and the resulting posterior draws could not be interpreted as samples from a coherent distribution. The key property in \citet{fong_martingale_2023} is that the sequences are c.i.d. as studied in \cite{berti2004limit}. We use the theory for almost conditionally identically distributed (a.c.i.d.) sequences established in \cite{battiston2025bayesianpredictiveinferencemartingales} to extend the guarantee to the interventional distributions targeted by our algorithm.




\subsection{Practical considerations}

\paragraph{Choosing the predictive rule.} An important question is how to choose the predictive rule used to construct the predictive sequences. The framework is permissive here: in principle any valid predictive rule may be used, and because the joint predictive factorises into separate components, each factor can be matched to its role in the analysis. The conditional outcome model $p(y \mid x, w)$ enters the interventional distribution directly through \Cref{eq:identification}, so flexibility matters most here. For continuous outcomes, we therefore recommend the conditional copula regression of \citet[][Section 4.4]{fong_martingale_2023}, which adapts to complicated conditional relationships without assuming a parametric likelihood. When the propensity $p(x \mid w)$ is modelled explicitly, a logistic-regression plug-in predictive is a convenient and interpretable choice for binary treatments. For the covariate distribution $p(w)$, we default to the Bayesian bootstrap, which requires no tuning and we adopt it even for purely discrete covariates, for simplicity. As a fully nonparametric alternative, the Bayesian bootstrap can instead be applied jointly to $(y, x, w)$. We illustrate and compare several of these choices in the examples below.

\paragraph{Sequence length.} Ideally, the imputed population size $N$ should be chosen as large as computationally feasible, though a large $N$ may not be necessary if the density updates become negligible well before the final step. As proposed by \cite{fong_martingale_2023} and \cite{ng_tabmgp_2026}, we recommend monitoring the mean $L_1$ distance between the initial fit and the intermediate density $$\bar{L}_1(i) = \frac{1}{B} \sum_{b=1}^B \left\lVert p_n(y(x)) - p_i^{(b)}(y(x))\right\rVert_1$$ as a function of the forward steps $i=n+1, \ldots, N$. 
A suitable $N$ is one for which the mean $L_1$ distance has stabilised, indicating convergence of the predictive distributions. This diagnostic requires the causal marginalisation to be performed at each forward step rather than only for the final density, which adds computational overhead. We suggest running a short pilot with a small number of predictive sequences to calibrate $N$ before the full analysis.

\paragraph{S-Learner versus T-Learner?} By default, we treat the treatment assignment $x$ as an additional covariate within a single model when updating the conditional model $p(y \mid x, w)$, which we refer to as an S-Learner\footnote{We use this terminology borrowing from the literature on heterogeneous treatment effects, where an S-Learner jointly learns the response surface for control and treatment group, whereas a T-Learner learns two separate response surfaces. See for example \cite{Kunzel_etal_19}, \cite{hahn_bayesian_2020} or \cite{caron_estimating_2022}.} approach. This strategy is advantageous as it allows the model to share information across different treatment levels and naturally accommodates continuous treatment variables. Under the copula update, a new observation with covariates $(x_{i+1}, w_{i+1})$ updates the predictive density $p_i(y \mid x, w)$ at every $(x, w)$, but with a weight that decays as $(x, w)$ moves away from $(x_{i+1}, w_{i+1})$. Consequently, a control observation still informs the treated predictive whenever its covariate vector is similar, with the degree of information sharing controlled by the corresponding bandwidth. In contrast, one could also take a T-Learner approach that partitions the data by treatment level and fits entirely separate models. While the T-Learner offers greater functional flexibility by not imposing a joint structure across groups, it precludes information sharing between treatment levels and becomes significantly more computationally intensive as the number of treatment categories increases. We will contrast these two approaches in examples below. In more complicated settings with many treatment levels, we suspect that the S-Learner approach is more suitable. 

\paragraph{Distribution functions. } Our exposition focuses on probability density functions (pdfs), but our method can easily target cumulative distribution functions (cdfs) as well. In fact, when using the copula update for a continuous outcome as defined by \Cref{eq:copula_update}, we obtain cdfs as a byproduct of the algorithm. The causal identification via marginalisation over the covariate distribution applies analogously to cdfs. Quantiles can be obtained by numerically inverting the resulting predictive cdfs. Thus, our framework is more flexible than most existing distributional causal inference methods, which are typically restricted to targeting only one of pdfs, cdfs, or quantiles.

\paragraph{Computational details.} Our method inherits the computational benefits of the martingale posterior framework, most notably that the predictive recursions are entirely parallel. The only additional cost is the causal marginalisation, which is an $\mathcal{O}(N)$ operation. For the copula recursion, this is preceded by $\mathcal{O}(n^2)$ and $\mathcal{O}(n)$ operations for fitting the copula density and $\mathcal{O}(N)$ operations for the predictive resampling. The $\mathcal{O}(N)$ operations are fully parallelisable across predictive sequences. In our empirical examples, the algorithm runs in no more than a few minutes, making it very fast compared to other Bayesian methods. Our JAX implementation, which is based on the original code from \cite{fong_martingale_2023}, can benefit from substantial speedups when run on a GPU. The code to reproduce our results is available at \url{https://github.com/gregorsteiner/CausalMP}.

\subsection{Conditional counterfactual distributions}

The approach outlined above can easily be extended to target conditional interventional distributions, opening up more specific comparisons. For example, a common estimand of interest is the average treatment effect on the treated (ATT), which compares mean counterfactual outcomes specifically for individuals who received the treatment. Under the same assumptions as above, the control outcome distribution conditional on receiving the treatment can be identified as
\begin{align*}
    p(y(0) \mid X = 1) = \int p(y(0) \mid X = 1, w) p(w \mid X = 1) \mathrm{d}w = \int p(y \mid X = 0, w) p(w \mid X = 1) \mathrm{d}w.
\end{align*}
In fact, assuming unconfoundedness only for the control outcome, i.e. $Y(0) \ind X \mid W$, is sufficient here. This modified adjustment can be easily integrated into the martingale posterior framework by integrating over the covariate distribution in the treatment group. More precisely, we approximate the integral as
\begin{align} \label{eq:identification_att}
    p_N(y(0) \mid X = 1) \approx \frac{1}{\sum_{i=1}^N x_i} \sum_{i: x_i = 1} p_N(y \mid X = 0, W = w_i),
\end{align}
where $\{x_i, w_i\}_{i=1}^N$ are resampled values of the treatment assignment and covariates. These ideas can be applied analogously to obtain counterfactual distributions in the  control group.


We can also condition directly on specific covariates to obtain subgroup comparisons. Under conditional ignorability, the conditional counterfactual distribution is $p(y(x) \mid w) = p(y \mid x, w)$, which we obtain as a byproduct of our default algorithm. Instead of integrating over the full empirical distribution of $W$ in the post-processing step, one can simply fix the covariates at chosen values of interest to yield the desired conditional counterfactual distribution. Finally, we can condition on both treatment assignment and specific covariates at the same time. For instance, this can yield counterfactual distributions among the treated for specific subsamples, similar in spirit to the conditional average treatment effect among the treated (CATT) estimand. 

\subsection{Illustration with simulated data}

We adopt a setting similar to Simulation~1 in \citet{xu2022bayesiansemiparametricmethodestimating} with five continuous covariates $W_j \sim \mathrm{U}(-2,2)$, $j = 1,\dots,5$, of which the first $J = 2$ are confounders, \begin{equation*} X \mid W \sim \mathrm{Bernoulli}\!\left( \operatorname{expit}\!\Big( \tfrac{4}{J}\textstyle\sum_{j=1}^{J} W_j \Big) \right), \end{equation*} where $\operatorname{expit}(x)=\frac{\exp{(x)}}{1+\exp{(x)}}$. The remaining three covariates do not affect treatment but contribute to the variability of the outcome. The control potential outcome, shared across both scenarios, is generated as \begin{equation*} Y(0) \mid W = -2.3 + U_1 + U_1^2 + \varepsilon_0, \quad \varepsilon_0 \sim 0.75\,\mathrm{N}_{+}(0, 0.9^2) + 0.25\,\mathrm{N}_{-}(0, 0.3^2), \end{equation*} where $U_k = \operatorname{expit}\!\big( 0.8 \sum_{j=1}^{5} W_j + 0.1 \sum_{j=1}^{5} |W_j|^k \big)$ and $\mathrm{N}_{+}$, $\mathrm{N}_{-}$ denote half-normal distributions truncated to the positive and negative axes, respectively, so that $Y(0)$ is right-skewed. We consider two scenarios that differ only in the treated potential outcome. In Scenario~1, \begin{equation*} Y(1) \mid W = -1.3 + U_2 + U_2^2 + \varepsilon_1, \quad \varepsilon_1 \sim 0.75\,\mathrm{N}_{-}(0, 0.9^2) + 0.25\,\mathrm{N}_{+}(0, 0.3^2), \end{equation*} which mirrors the control distribution and is left-skewed, so that the two potential-outcome densities are broadly similar in shape, but have different skew. In Scenario~2, \begin{equation*} Y(1) \mid W \sim 0.7\,\mathrm{N}\!\left(-2.5 + 5 U_2,\, 0.35^2\right) + 0.3\,\mathrm{N}\!\left(2.5 - 5 U_2,\, 0.35^2\right), \end{equation*} which is bimodal and presents a more challenging estimation target.

We compare three variants of our method. All three model the outcome through the copula-based conditional density regression and target the marginal counterfactual densities by integrating over the covariate distribution. The first, an \emph{S-Learner with a Bayesian bootstrap treatment model}, fits a single conditional density on the full sample and propagates the joint distribution of $(X, W)$ via the Bayesian bootstrap in the forward sampling. The second, an \emph{S-Learner with a logistic treatment model}, retains the single outcome model but, rather than bootstrapping the treatment, draws $X$ from a logistic propensity model whose coefficients are updated recursively along the predictive sequence via a natural-gradient step, while $W$ is resampled by the Bayesian bootstrap. The third, a \emph{T-Learner with the Bayesian bootstrap}, fits a separate conditional density on each treatment arm and integrates over the covariate distribution drawn by the Bayesian bootstrap. For each variant, we generate $B = 100$ posterior samples, each propagated with $2{,}000$ forward samples, 
from which we report the posterior mean density and pointwise $95\%$ credible bands. To assess the stability of the forward sampling, we also compute for each posterior sample the $L_1$ distance between the initial fitted density and the intermediate density at every forward step, yielding a trajectory that shows how quickly the density stabilises. The ground truth is obtained by Monte Carlo integration of the conditional densities over the covariate distribution. \Cref{fig:sim_density} illustrates the estimated marginal counterfactual densities and the $L_1$ stability trajectories on a single simulated dataset for both scenarios.

In all settings, the $L_1$ distances stabilise well, suggesting that the predictive distributions converge to a well-defined limit. In Scenario 1, the S-Learners approximate the true counterfactual distributions well, whereas the T-Learner is slightly overconfident for the control outcome. In Scenario 2, where the treatment counterfactual distribution is bimodal, the T-Learner yields narrower credible intervals, especially for the bimodal distribution. The choice between logistic treatment resampling and the Bayesian bootstrap has negligible impact on the resulting posteriors in both scenarios.

\begin{figure}[h]
    \centering
    \includegraphics[width=\linewidth]{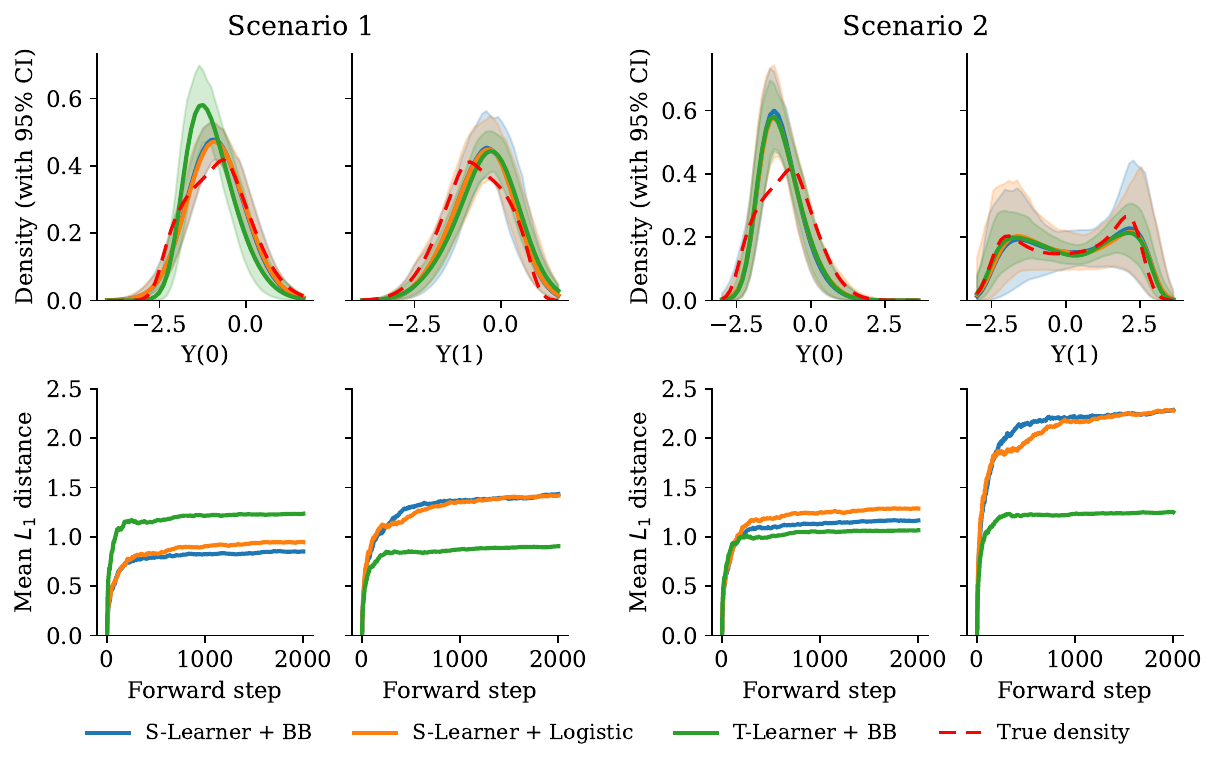}
    \caption{{\bf Simulated data.} Marginal counterfactual density estimates and $L_1$ stability diagnostics for Scenarios~1 and~2. Top: posterior mean densities with $95\%$ credible bands against the true density. Bottom: mean $L_1$ distance from the initial fitted density across forward sampling steps.}
    \label{fig:sim_density}
\end{figure}

\section{Extension to instrumental variables}

Our proposed methodology can be extended to instrumental variable designs with unobserved confounding. In the main text, we focus on the setting with a binary treatment and instrument. We sketch an extension to more general designs with continuous treatments and arbitrary instruments in \Cref{sec:iv_extension}, but leave a more detailed treatment for future work.

\subsection{A binary experiment with non-compliance} \label{sec:binary_noncompliance}

We observe data from a binary experiment $\mathcal{D}_{1:n} = \left\{ (y_i, x_i, z_i) \right\}_{i=1}^n$ where $z_i \in \{0, 1\}$ is the assigned treatment and $x_i \in \{0, 1\}$ is the treatment received. We consider latent potential outcomes $Y(X, Z)$ and potential treatments $X(Z)$. The variable $Z$ indicates assignment to the treatment and is often called an instrumental variable (IV). Key assumptions are  that the instrument is randomly assigned and does not directly affect the outcome.

\begin{assumption}[Random Assignment] \label{ass:iv_random}
    The instrument $Z$ is randomly assigned, $Z \ind (Y(x, z), X(z))$ for all $x, z = 0, 1$.
\end{assumption}

\begin{assumption}[Exclusion Restriction]
    The instrument affects the outcome only through its effect on the treatment, that is, $Y(X, Z)$ is constant in $Z$ such that we can write $Y(X) = Y(X, Z)$.
\end{assumption}

Individuals can be classified into four response types based on their latent potential treatments \citep[see e.g.][]{angrist_identification_1996}: Always-takers ($X(1) = X(0) = 1$), Never-takers ($X(1) = X(0) = 0$), Compliers ($X(1) = 1, X(0) = 0$), and Defiers ($X(1) = 0, X(0) = 1$). The response-type probabilities, denoted by $(p_{AT}, p_{NT}, p_{CP}, p_{DF})$, can be expressed in terms of the observed data as
\begin{align*}
    p(X = 0 \mid Z = 0) &= p_{NT} + p_{CP} \\
    p(X = 1 \mid Z = 0) &= p_{AT} + p_{DF} \\
    p(X = 0 \mid Z = 1) &= p_{NT} + p_{DF} \\
    p(X = 1 \mid Z = 1) &= p_{AT} + p_{CP}.
\end{align*}
Finally, we assume monotonicity, which rules out the existence of defiers ($p_{DF} = 0$).
\begin{assumption}[Monotonicity] \label{ass:iv_monotonicity}
    The treatment assignment is monotone in the sense that $X(1) \geq X(0)$ with probability $1$ and $X(1) > X(0)$ with positive probability.
\end{assumption}
Under the monotonicity assumption, the remaining three response-type probabilities can be identified from the observed data by solving the system of equations above. The strict version adopted here guarantees that the proportion of compliers is non-zero. Then, the interventional distributions for the complier population, denoted by $p(y(x) \mid CP)$, can be identified as \citep{imbens1997estimating}
\begin{align} \label{eq:late_identification}
\begin{aligned}
    p(y(0) \mid CP) &= \frac{p_{NT} + p_{CP}}{p_{CP}} p(y \mid X = 0, Z = 0) - \frac{p_{NT}}{p_{CP}} p(y \mid X = 0, Z = 1) \\
    p(y(1) \mid CP) &= \frac{p_{AT} + p_{CP}}{p_{CP}} p(y \mid X = 1, Z = 1) - \frac{p_{AT}}{p_{CP}} p(y \mid X = 1, Z = 0).
\end{aligned}
\end{align}
We propose to resample $(X, Z)$ pairs using the Bayesian bootstrap and to learn the distribution $p(y \mid x, z)$ via the conditional copula update. For each predictive sequence, the response-type probabilities can be estimated from the resampled $(X, Z)$ pairs and the interventional distributions can be computed via \Cref{eq:late_identification}.

The validity of this construction again hinges on the sequence of complier interventional densities converging to a well-defined limit. Unlike the setting of \Cref{thm:convergence}, the complier density in \Cref{eq:late_identification} is a ratio of reduced-form quantities, whose numerator and denominator both evolve along the predictive sequence. The following result, proved in \Cref{app:proof_iv_convergence}, shows that convergence nonetheless holds provided the complier share is bounded away from zero.

\begin{theorem}[Convergence of the complier interventional density] \label{thm:iv_convergence}
    Suppose the predictive update for $p_i(y \mid x, z)$ satisfies the martingale property with a uniformly bounded outcome density, and that $X$ and $Z$ are updated by a Bayesian bootstrap. If Assumption \ref{ass:iv_monotonicity} holds, then for each $x \in \{0, 1\}$ there exists a random probability measure $P_\infty(y(x) \mid CP)$ such that $P_i(y(x) \mid CP)$ converges weakly to $P_\infty(y(x) \mid CP)$ almost surely as $i \to \infty$.
\end{theorem}

\subsection{Extension to covariates}

The binary instrumental variable approach extends naturally to settings where covariates $W$ are available, which is useful when the instrument is only conditionally randomly assigned, or when conditional complier effects are of interest. Under the conditional versions of Assumptions~\ref{ass:iv_random}--\ref{ass:iv_monotonicity}, the conditional complier interventional distributions $p(y(x) \mid CP, w)$ are identified by the covariate-conditional analogue of Equation~\eqref{eq:late_identification}, replacing each observed conditional outcome distribution and response-type probability by its conditional counterpart
\begin{align} \label{eq:conditional_late_identification}
\begin{aligned}
    p(y(0) \mid CP, w) &= \frac{p(NT \mid w) + p(CP \mid w)}{p(CP \mid w)} p(y \mid X = 0, Z = 0, W = w) \\
    &- \frac{p(NT \mid w)}{p(CP \mid w)} p(y \mid X = 0, Z = 1, W = w) \\
    p(y(1) \mid CP, w) &= \frac{p(AT \mid w) + p(CP \mid w)}{p(CP \mid w)} p(y \mid X = 1, Z = 1, W = w) \\
    &- \frac{p(AT \mid w)}{p(CP \mid w)} p(y \mid X = 1, Z = 0, W = w).
\end{aligned}
\end{align}
The marginal complier distribution is then recovered by integrating over the covariate distribution among compliers,
\begin{equation}
  p(y(x) \mid CP) = \int p(y(x) \mid CP, w)\, p(w \mid CP)\, \mathrm{d}w,
  \label{eq:complier-marginal}
\end{equation}
where the covariate distribution among compliers follows from Bayes' theorem as
\begin{equation*}
    p(w \mid CP) = \frac{p(CP \mid w)\, p(w)}{p_{CP}},
\end{equation*}
with $p(CP \mid w) = p(X = 1 \mid Z = 1, W = w) - p(X = 1 \mid Z = 0, W = w)$ and $p_{CP} = \int p(CP\mid w)\,p(w)\,\mathrm{d}w$ the unconditional complier share.

In \Cref{eq:conditional_late_identification}, each term carries a factor $1/p(CP\mid w)$, which is cancelled by the weight $p(CP\mid w)$ implicit in \eqref{eq:complier-marginal}. For the control outcome (and analogously for treatment), this leaves
\begin{equation*}
  p(y(0)\mid CP)
  = \frac{\int \left[p(y, X = 0 \mid Z = 0, w) - p(y, X = 0 \mid Z = 1, w) \right] p(w)\, \mathrm{d}w}{\int p(CP\mid w)\, p(w)\, \mathrm{d}w},
\end{equation*}
a ratio of reduced-form quantities. This integrates well into the martingale posterior framework: we may resample $(Y, X, Z, W)$, for instance, by modelling $p(y \mid x, z, w)$ with the conditional copula update, $p(x \mid z, w)$ and $p(z \mid w)$ with a parametric update, and $p(w)$ with the Bayesian bootstrap. Within each predictive sequence, the conditional complier probability is read off from the predictive treatment model as 
$$p_N(CP \mid w_i) = p_N(X=1 \mid Z=1, W=w_i) - p_N(X=1 \mid Z=0, W=w_i),$$
and the marginal complier distribution is approximated by the weighted average
$$p_N(y(x) \mid CP) \approx \frac{\sum_{i=1}^{N} p_N(CP \mid w_i)\, p_N(y(x) \mid CP, w_i)}{ \sum_{i=1}^{N} p_N(CP \mid w_i)},
$$
where $\{w_i\}_{i=1}^N$ are resampled from the marginal covariate distribution. Fixing $w$ rather than averaging yields conditional complier effects.

Our construction requires strict monotonicity to hold within each covariate stratum, that is, \Cref{ass:iv_monotonicity} holds conditional on every value of $w$. This guarantees that the conditional complier share $p(CP \mid w)$ is bounded away from zero, and under this condition the convergence guarantee of \Cref{thm:iv_convergence} carries over to the setting with covariates. 

The conditional monotonicity that our construction requires can be a strong assumption in practice. It is worth noting, however, that traditional linear IV estimators such as two-stage least squares (TSLS) rely on it too, and can lose their causal interpretation when the direction of monotonicity varies with the covariates \citep{sloczynski2026}. The estimand then becomes a weighted average of covariate-specific complier average effects in which some of the weights may be negative. A similar problem arises when the covariates are not controlled for in a sufficiently ``rich'' fashion \citep{blandhol2025}, in which case the estimand may depend on the levels of the potential outcomes rather than on treatment effects alone, and can no longer be expressed as a non-negatively weighted average of subgroup effects. Our approach avoids this latter problem by construction: the identification in \eqref{eq:conditional_late_identification} operates within each covariate stratum directly, without forming a linear projection of $Z$ onto $W$ that could be misspecified.


\section{Examples with real data}

\subsection{The effect of zinc supplementation on the duration of the common cold}

We illustrate our methodology using data from several randomised, double-blind, placebo-controlled trials investigating the effect of zinc acetate lozenges on common cold duration, as previously analysed by \cite{hemila_estimating_2025}. Looking at average treatment effects can be misleading because the absolute benefit of treatment is expected to vary depending on the severity of the illness.

To conduct our inference, we generate martingale posterior samples for the potential outcome distributions using the approach suggested in \Cref{sec:mp_counterfactual}. As the data are from clinical trials, the treatment assignment can be assumed to be unconfounded even without conditioning on covariates. We compare the S-Learner approach that jointly updates treatment and control data to the T-Learner approach that learns them separately. To formally quantify how the treatment effect varies across the outcome distribution, we compute quantile treatment effects: for each posterior draw, we invert the modelled cumulative distribution functions of the treated and control potential outcomes to obtain the corresponding quantiles, and take their difference. This yields, for any quantile level $\delta$, a posterior distribution of the $\delta$-quantile treatment effect $Q_1(\delta) - Q_0(\delta)$.

\Cref{fig:zn_example} displays the estimated counterfactual densities and quantile treatment effects. The treated distribution is centred on lower durations, but is also more concentrated. In particular, the control distribution has heavier tails, suggesting that individuals who would otherwise suffer from the longest-lasting colds derive the greatest absolute benefit from zinc supplementation. The S- and T-Learner estimates are broadly consistent, with the main difference being slightly wider credible intervals for the control distribution under the T-Learner. This pattern is confirmed by the posterior quantile treatment effects shown in the bottom row: the posterior mean effect of zinc supplementation grows from a reduction of $1.4$ days at the $0.1$-quantile to around $3.8$ days at the $0.9$-quantile, for both learners, indicating that patients with more severe colds benefit most from treatment.

\begin{figure}
    \centering
    \includegraphics[width=0.9\linewidth]{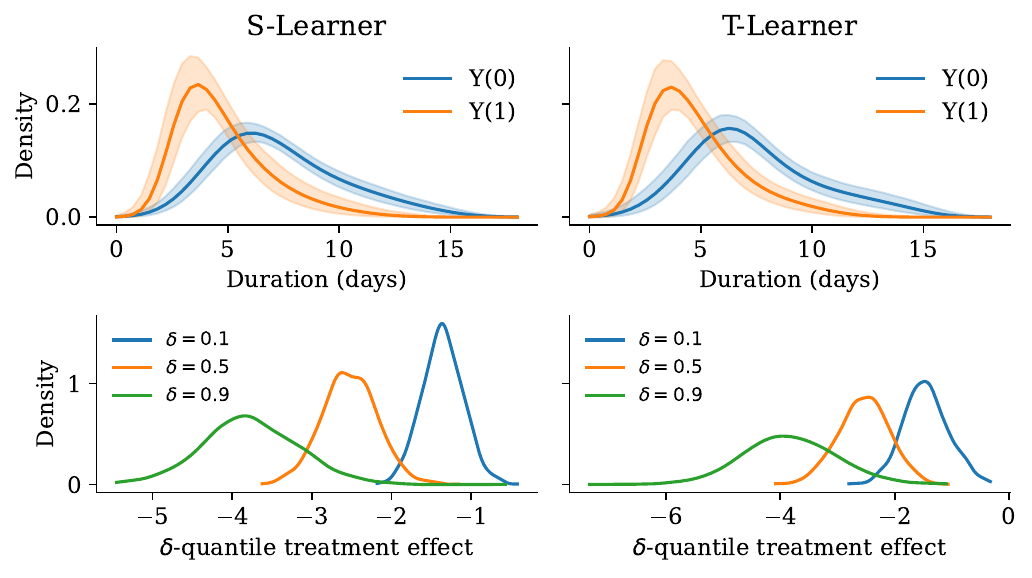}
    \caption{{\bf Zinc supplementation data.} Top row: estimated counterfactual densities (with $95\%$ CIs) for the duration of the common cold under zinc supplementation and control, for the S-Learner (left) and T-Learner (right). Bottom row: posterior distributions of the $\delta$-quantile treatment effect $Q_1(\delta) - Q_0(\delta)$ for $\delta \in \{0.1, 0.5, 0.9\}$. Results are based on $1{,}000$ predictive sequences of $5{,}000$ forward samples each.}
    \label{fig:zn_example}
\end{figure}

\subsection{A vitamin A supplementation trial with one-sided non-compliance}

We illustrate the martingale posterior approach on a dataset investigating the impact of vitamin A supplementation on children's survival rates \citep{sommer_estimating_1991}. Villages in Indonesia were randomly assigned to receive vitamin supplements. The experiment exhibits one-sided non-compliance: while no individuals in the control villages had access to the supplements, a subset of those in the treatment villages failed to receive them. Thus, the monotonicity assumption holds by design and all individuals are either compliers or never-takers. Table~\ref{tab:sommer_zeger_data} shows a contingency table for the dataset.

\begin{table}[h]
\centering
\begin{tabular}{lcccc}
\hline
\textbf{Response Type} & \textbf{Assigned ($Z$)} & \textbf{Received ($X$)} & \textbf{Survived ($Y$)} & \textbf{Count} \\ \hline
Complier or Never-taker    & 0                      & 0                      & 0              & 74                   \\
Complier or Never-taker    & 0                      & 0                      & 1           & 11,514               \\
Never-taker                & 1                      & 0                      & 0              & 34                   \\
Never-taker                & 1                      & 0                      & 1          & 2,385                \\
Complier                   & 1                      & 1                      & 0               & 12                   \\
Complier                   & 1                      & 1                      & 1           & 9,663                \\ \hline
\end{tabular}
\caption{Sommer-Zeger Vitamin A Dataset \citep[Table taken from][]{imbens_bayesian_1997}}
\label{tab:sommer_zeger_data}
\end{table}

\begin{figure}[h]
    \centering
    \includegraphics[width=0.9\linewidth]{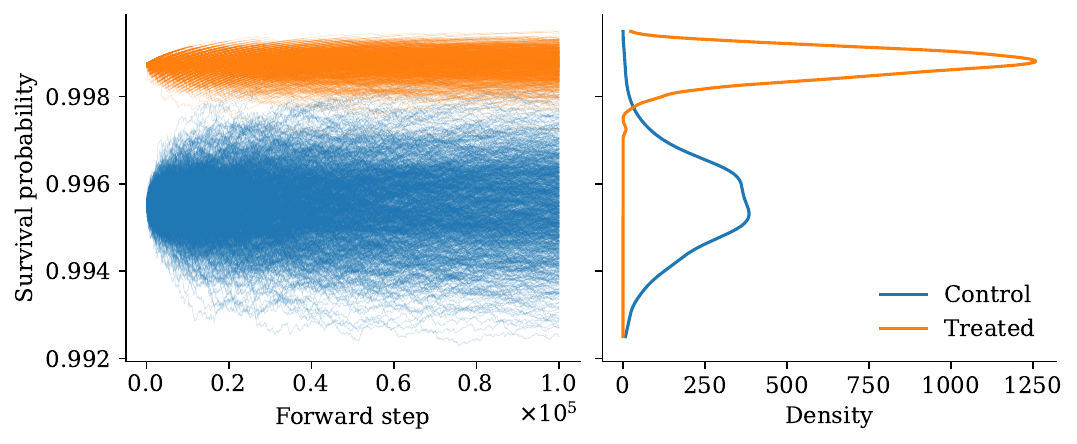}
    \caption{{\bf Vitamin A data.}  Martingale posterior predictive sequences of survival probabilities for the control and treatment (complier) counterfactual distributions (left), with the corresponding posterior densities of the final estimates (right). Results are based on $1{,}000$ predictive sequences of $100{,}000$ forward samples generated with the Bayesian bootstrap.
    }
    \label{fig:results_sommer_zeger}
\end{figure}

We apply our procedure based on the Bayesian bootstrap to obtain complier counterfactual distributions, which are fully characterised by survival probabilities in the binary outcome setting. We display martingale posteriors of the complier survival probabilities in Figure~\ref{fig:results_sommer_zeger}. A martingale posterior of the average treatment effect among compliers is implied by taking the difference of the survival probabilities in each posterior sample, and gives a $90\%$ credible interval of $[1.47, 4.94]$ (in increased survival per $1{,}000$ individuals). For the proportion of individuals who comply with the treatment assignment, we obtain a $90\%$ credible interval of $[0.795, 0.806]$. Our results are consistent with, but slightly more concentrated than those in \cite{imbens_bayesian_1997} who obtain a $90\%$ credible interval of $[1.2, 5.1]$ for the complier treatment effect. In addition, our approach is conceptually simpler and computationally fast.

\subsection{LaLonde: The effect of job training}

We reanalyse the experimental job training data from \cite{lalonde_evaluating_1986}, specifically the subset reconstructed by \cite{dehejia1999causal} containing $185$ treated and $260$ control observations. Because the program targeted individuals with particularly poor employment prospects, the resulting treatment effects cannot be easily generalised to the broader population. Thus, the estimand of choice for this dataset is typically the average treatment effect among the treated (ATT), where a weaker version of the overlap assumption is sufficient\footnote{We only need $p(X = 0 \mid w) > 0$ for all values of $w$ with $p(w \mid X = 1) > 0$. In words, for any covariate stratum in the treatment group, there needs to be positive probability of such an individual appearing in the control group. This is necessary for being able to learn about $p_N(y|X=0, W=w_i)$ in (\ref{eq:identification_att}). However, the converse does not need to hold, so the control covariate distribution can be more diverse.}. A comprehensive discussion of the application and the methods used in \cite{lalonde_evaluating_1986} as well as its impact on the current state of the field is provided in \cite{ImbensXu25}. 

We learn the counterfactual control distribution for the treated $Y(0) \mid X = 1$ as described in \Cref{eq:identification_att} and compare this with the distribution of $Y(1) \mid X = 1$, as well as the implied ATT. The latter density is trivially identified and can be obtained by analogously evaluating the conditional density at $X=1$. We compare two versions of the method, which differ in how the treatment is resampled. In the first, the treatment and covariates are drawn jointly from the Bayesian bootstrap. In the second, only the covariates are drawn from the Bayesian bootstrap, while the treatment is resampled from a logistic regression whose parameters are updated recursively. The second version quantifies uncertainty in the propensity score automatically, through the martingale posterior on the logistic-regression parameters. The first does not model the propensity explicitly, so obtaining the same uncertainty quantification requires fitting a separate logistic regression model to the imputed observations. The outcome variable, real earnings in 1978, has a large point mass at zero corresponding to participants who remained unemployed, which is inconsistent with the continuous copula-regression update. We therefore model $Y$ as a zero-inflated mixture: the point mass $P(Y = 0 \mid x, w)$ is estimated by a logistic regression whose coefficients are updated recursively, giving a martingale posterior over the atom probability, while the continuous part $p(y \mid Y \neq 0, x, w)$ is estimated with the copula regression update as before, fitted only on the non-zero observations. For the ATT inference, the two components are combined by taking, for each posterior draw, $\mathbb{E}[Y(x)] = (1 - \pi_0(x)) \, \mathbb{E}[Y(x) \mid Y(x) \neq 0]$, where $\pi_0(x)$ denotes the mixture weight for $Y(x) = 0$ at treatment level $x = 0, 1$.

\begin{figure}[h]
    \centering
    \includegraphics[width=\linewidth]{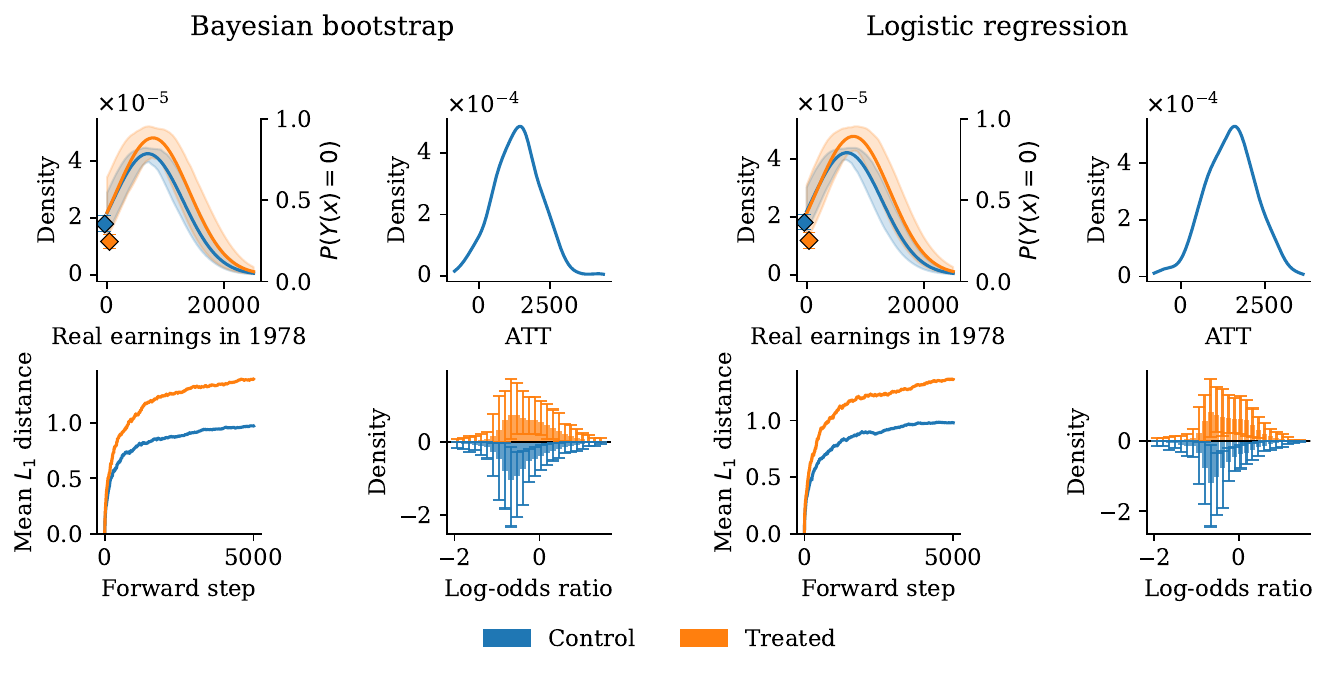}
    \caption{{\bf LaLonde job training data.} Bayesian bootstrap (left) versus Logistic treatment resampling (right). Top row: posterior mean counterfactual densities for the treated with 95\% credible intervals, alongside the posterior point mass $p(Y(x)=0)$ on the secondary axis (left), and posterior distribution of the ATT (right). Bottom row: mean $L_1$ convergence diagnostic for the continuous part (left) and posterior propensity score distributions on the log-odds scale (right). The results are based on $200$ predictive sequences of $5{,}000$ forward samples each.}
    \label{fig:lalonde_results}
\end{figure}

%
\Cref{fig:lalonde_results} displays the estimated counterfactual distributions and the implied martingale posterior on the ATT, along with the $L_1$ convergence diagnostic and log-odds ratios of the estimated propensity scores for the two predictive schemes. For both variants, the treated counterfactual distribution places slightly more mass on higher earnings, though the credible intervals overlap for the most part. The point mass at zero, however, is notably lower for the treated, suggesting that the job training programme helps participants find employment. For the implied ATT, we obtain $95\%$ credible intervals of $[-200, 2818]$ for the Bayesian bootstrap and $[135, 2866]$ for the logistic update, broadly in line with, though somewhat wider than, the experimental estimates in \cite{ImbensXu25}. As expected for an experimental sample, the propensity scores overlap closely between treatment groups, implying the required overlap in covariate distributions \citep{ImbensXu25}. The mean $L_1$ distances for the continuous part of the outcome distribution indicate that both predictive updates stabilise well.

\section{Conclusion}

We propose a martingale-posterior approach to inference on causal counterfactual distributions. Rather than targeting a single summary such as the average treatment effect, our method estimates entire counterfactual outcome distributions, from which any derived functional can be recovered, each accompanied by coherent epistemic uncertainty quantification obtained directly from the martingale posterior samples. The construction relies on predictive resampling with flexible predictive rules, so it inherits a robustness to restrictive parametric assumptions. We first develop the methodology under standard unconfoundedness and then extend it to instrumental-variable designs that accommodate unobserved confounding, covering both marginal and conditional counterfactual distributions. On the theoretical side, we establish convergence of the underlying predictive recursions, and we illustrate the practical behaviour of the approach on several real datasets.

Several extensions are worth pursuing. Richer predictive rules could improve the quality of the counterfactual estimates. Foundation models for tabular data are a promising candidate: their autoregressive generation mirrors the predictive forward sampling, and \citet{ng_tabmgp_2026} show that a martingale posterior built on TabPFN \citep{hollmann2025accurate} performs well in many settings. Adapting such rules to our causal setting is a natural next step. A second direction is to consider influence function-based updates of the target functional, which could yield a predictive, martingale-posterior interpretation of influence function-based estimators in causal inference.

\section*{Acknowledgements}

We used large language models, in particular Anthropic's Claude Opus and Sonnet families, for coding assistance and proofreading. 
We are solely responsible for the content and any errors.

\bibliographystyle{apalike}
\bibliography{references}

\appendix

\section{Proofs}

\subsection{Proof of \Cref{thm:convergence}} \label{app:proof_convergence}

\paragraph{Setup.} Let $\mathcal{F}_i = \sigma(Y_{1:i}, X_{1:i}, W_{1:i})$ for $i \geq 1$ denote the filtration generated by the observed and imputed data. Write $P_i(y \mid x, w)$ and $P_i(w)$ for the random probability measures induced by the predictive rule at step $i$ for the conditional outcome distribution and the marginal covariate distribution respectively, with densities
\[
p_i(y \mid x, w) = p(y \mid x, w, \mathcal{F}_i), \qquad p_i(w) = p(w \mid \mathcal{F}_i)
\]
with respect to dominating measures $\mu_y$ (Lebesgue measure if $Y$ is continuous, a counting measure if discrete) and $\mu_w$ (a counting measure on the support of $W$ when $W$ is updated via the Bayesian bootstrap, Lebesgue measure otherwise). The interventional predictive density at step $i$ is, as in \Cref{eq:identification},
\begin{align} \label{eq:interventional_density_i}
    p_i(y(x)) := \int p_i(y \mid x, w)\, p_i(w)\, \mathrm{d}\mu_w(w).
\end{align}

The proof requires that each predictive rule individually satisfies the pointwise martingale property of \Cref{sec:predictive_rules}, that is, $\mathbb{E}[p_i(y \mid x, w) \mid \mathcal{F}_{i-1}] = p_{i-1}(y \mid x, w)$ for all $y, x, w$ and $\mathbb{E}[p_i(w) \mid \mathcal{F}_{i-1}] = p_{i-1}(w)$ for every $i \geq n+1$. This holds, for instance, when $p_i(y \mid x, w)$ follows the covariate-dependent copula update of \Cref{eq:copula_update} and $p_i(w)$ is updated via the Bayesian bootstrap. In addition, \Cref{thm:convergence} requires the boundedness conditions $\sup_i \sup_{y,x,w} p_i(y \mid x, w) \leq \bar{p} < \infty$ a.s.\ and $\sup_{u,v \in (0,1)} c_\rho(u,v) < \infty$ for the copula densities used to update $p_i(y \mid x, w)$.

\begin{proof}[Proof of \Cref{thm:convergence}]
Fix the treatment level $x$ and $i \geq n+1$. Conditioning on $\mathcal{F}_{i-1}$,
\begin{align} \label{eq:fubini_step}
    \mathbb{E}\!\left[p_i(y(x)) \mid \mathcal{F}_{i-1}\right] = \mathbb{E}\!\left[ \int p_i(y \mid x, w)\, p_i(w)\, \mathrm{d}\mu_w(w) \;\middle|\; \mathcal{F}_{i-1} \right] = \int \mathbb{E}\!\left[ p_i(y \mid x, w)\, p_i(w) \mid \mathcal{F}_{i-1} \right] \mathrm{d}\mu_w(w),
\end{align}
where the second equality is justified by Tonelli's theorem, which requires that the integrand $p_i(y \mid x, w)\, p_i(w)$ be non-negative and jointly measurable.

Write $p_i(y \mid x, w) = p_{i-1}(y \mid x, w) + \Delta_i(y \mid x, w)$ and $p_i(w) = p_{i-1}(w) + \Delta_i(w)$, where by the individual martingale property, $\mathbb{E}[\Delta_i(y \mid x, w) \mid \mathcal{F}_{i-1}] = 0$ pointwise for fixed $w$ and $x$ and $\mathbb{E}[\Delta_i(w) \mid \mathcal{F}_{i-1}] = 0$. Expanding the product, $p_{i-1}(y \mid x, w)$ and $p_{i-1}(w)$ are $\mathcal{F}_{i-1}$-measurable so that the cross-terms are zero, leaving
\[
\mathbb{E}\!\left[p_i(y \mid x, w)\, p_i(w) \mid \mathcal{F}_{i-1}\right] = p_{i-1}(y \mid x, w)\, p_{i-1}(w) + \mathbb{E}\!\left[\Delta_i(y \mid x, w)\, \Delta_i(w) \mid \mathcal{F}_{i-1}\right].
\]
Thus, we can rewrite \eqref{eq:fubini_step} as
\begin{align} \label{eq:covariance_term}
    \mathbb{E}\!\left[p_i(y(x)) \mid \mathcal{F}_{i-1}\right] - p_{i-1}(y(x)) = \int \mathbb{E}\!\left[ \Delta_i(y \mid x, w)\, \Delta_i(w) \mid \mathcal{F}_{i-1} \right] \mathrm{d}\mu_w(w) =: \gamma_i(x).
\end{align}
The term $\gamma_i(x)$ is the (integrated) conditional covariance between the two increments. The interventional density $p_i(y(x))$ is itself a martingale if and only if $\gamma_i(x) = 0$ for all $i$, which need not hold in general because the increment of the outcome density depends on the newly drawn covariate value $w_i$. It remains to bound $\gamma_i(x)$ by a summable sequence so that Theorem 1 in \cite{battiston2025bayesianpredictiveinferencemartingales} applies.

We verify summability for the default predictive rules of \Cref{sec:predictive_rules}. Under the Bayesian bootstrap, $p_i(w) = \frac{i-1}{i} p_{i-1}(w) + \frac{1}{i} 1\{w = w_i\}$, so
\begin{align} \label{eq:delta_w}
    \Delta_i(w) = \frac{1}{i}\left(1\{w = w_i\} - p_{i-1}(w)\right).
\end{align}
We can bound the absolute increments summed over the atoms of its support by
\begin{align} \label{eq:tv_bound}
    \int |\Delta_i(w)|\, \mathrm{d}\mu_w(w) = \frac{1}{i} \left[ \left(1 - p_{i-1}(w_i)\right) + \sum_{w \neq w_i} p_{i-1}(w) \right] = \frac{2}{i}\left(1 - p_{i-1}(w_i)\right) \leq \frac{2}{i}.
\end{align}
Under the copula update, the increment of the outcome density is
\[
    \Delta_i(y \mid x, w) = \alpha_i([x, w], [x_i, w_i]) \left( c_{\rho_y}\!\left(P_{i-1}(y \mid x, w), P_{i-1}(y_i \mid x_i, w_i)\right) - 1\right) p_{i-1}(y \mid x, w),
\]
where $\alpha_i([x, w], [x_i, w_i])$ is the covariate-dependent weight defined analogously to the regression extension in \Cref{sec:predictive_rules}, applied to the joint covariate vector $(x, w)$. Provided the bandwidths $\rho_x, \rho_y$ are bounded away from one, the Gaussian copula density is uniformly bounded, $\sup_{u,v \in (0,1)} c_\rho(u,v) =: \bar{c}_\rho < \infty$, and hence $\alpha_i([x, w], [x_i, w_i]) \leq \bar{\alpha}/i$ for a constant $\bar{\alpha} < \infty$ not depending on $(x, w, x_i, w_i)$. Combined with the boundedness condition $\sup_i \sup_{y,x,w} p_i(y \mid x, w) \leq \bar{p}$, this gives
\begin{align} \label{eq:delta_y_bound}
    \left|\Delta_i(y \mid x, w)\right| \leq \frac{\bar{\alpha}\, \bar{p}\, (\bar{c}_\rho \vee 1)}{i} =: \frac{\bar{\Delta}}{i} \qquad \text{a.s.}
\end{align}
Combining \eqref{eq:tv_bound} and \eqref{eq:delta_y_bound},
\[
    \gamma_i(x) \leq \int \left|\Delta_i(y \mid x, w)\right| \left|\Delta_i(w)\right| \mathrm{d}\mu_w(w) \leq \frac{\bar{\Delta}}{i} \int |\Delta_i(w)|\, \mathrm{d}\mu_w(w) \leq \frac{2\bar{\Delta}}{i^2} \qquad \text{a.s.},
\]
so $\xi_i = 2\bar{\Delta}/i^2$ satisfies $\sum_{i=n+1}^\infty \xi_i < \infty$. Thus $\mathbb{E}[p_i(y(x)) \mid \mathcal{F}_{i-1}] \leq p_{i-1}(y(x)) + \xi_i$ a.s.\ with $\sum_{i=n+1}^\infty \xi_i < \infty$, and by \citet[Theorem 1]{battiston2025bayesianpredictiveinferencemartingales} there exists a random probability measure $P_\infty(y(x))$ such that $P_i(y(x))$ converges weakly to $P_\infty(y(x))$ almost surely as $i \to \infty$, completing the proof.
\end{proof}

\begin{remark}
    \Cref{thm:convergence} is stated for the default predictive rules of \Cref{sec:predictive_rules}, but the argument is not tied to this particular choice: it applies to any pair of predictive rules for which similar summable rates can be established for the covariance term $\gamma_i(x)$. This includes, for example, the parametric plug-in update of \Cref{sec:predictive_rules}, whose increments are $\mathcal{O}(i^{-1})$ by construction.
\end{remark}

\subsection{Proof of \Cref{thm:iv_convergence}} \label{app:proof_iv_convergence}

We start by proving a lemma that guarantees that a product ratio of martingales satisfies the almost supermartingale property in \cite{battiston2025bayesianpredictiveinferencemartingales}.

\begin{lemma}
\label{lem:acid-ratio}
Let $\mathcal{F} = (\mathcal{F}_i)_{i \geq 1}$ be a filtration, and let $(a_i)_{i \geq 1}$, $(b_i)_{i \geq 1}$, $(c_i)_{i \geq 1}$ be $\mathcal{F}$-martingales such that:
\begin{enumerate}
    \item[(i)] there exist constants $A, B < \infty$ and $\varepsilon > 0$ with $|a_i| \leq A$, $|b_i| \leq B$, and $c_i \geq \varepsilon$ a.s.\ for all $i \geq 1$;
    \item[(ii)] there exist non-negative $\mathcal{F}$-adapted sequences $(\xi_i^a)_{i \geq 1}$, $(\xi_i^b)_{i \geq 1}$, $(\xi_i^c)_{i \geq 1}$, each a.s.\ summable, such that for all $i \geq 1$, a.s.,
    \begin{equation*}
        \mathbb{E}\big[ (a_{i+1} - a_i)^2 \,\big|\, \mathcal{F}_i \big] \leq \xi_i^a, \qquad
        \mathbb{E}\big[ (b_{i+1} - b_i)^2 \,\big|\, \mathcal{F}_i \big] \leq \xi_i^b, \qquad
        \mathbb{E}\big[ (c_{i+1} - c_i)^2 \,\big|\, \mathcal{F}_i \big] \leq \xi_i^c.
    \end{equation*}
\end{enumerate}
Then $Z_i := a_i b_i / c_i$ satisfies the (two-sided) almost-supermartingale bound: for all $i \geq 1$,
\begin{equation}
\label{eq:acid-ratio-bound}
    \big| \mathbb{E}\big[ Z_{i+1} - Z_i \,\big|\, \mathcal{F}_i \big] \big| \;\leq\; \xi_i \quad \text{a.s.},
    \qquad \text{where} \quad
    \xi_i := K \big( \xi_i^a + \xi_i^b + \xi_i^c \big), \quad
    K := \frac{1}{2\varepsilon} + \frac{3A + B}{2\varepsilon^2} + \frac{AB}{\varepsilon^3},
\end{equation}
and $\sum_{i=1}^{\infty} \xi_i < \infty$ a.s. In particular, $(Z_i)_{i \geq 1}$ is a $(\xi_i)_{i \geq 1}$-almost supermartingale with respect to $\mathcal{F}$.
\end{lemma}

\begin{proof}[Proof of Lemma~\ref{lem:acid-ratio}]
Fix $i \geq 1$ and write $\Delta a := a_{i+1} - a_i$, $\Delta b := b_{i+1} - b_i$, $\Delta c := c_{i+1} - c_i$. Expanding $a_{i+1} b_{i+1} = (a_i + \Delta a)(b_i + \Delta b)$ and $1/c_{i+1} = 1/c_i - \Delta c / c_i^2 + (\Delta c)^2 / (c_i^2 c_{i+1})$ yields the exact identity
\begin{equation*}
    Z_{i+1} - Z_i
    = \frac{b_i \Delta a + a_i \Delta b}{c_i} - \frac{a_i b_i}{c_i^2} \Delta c
    \;+\; \frac{\Delta a \Delta b}{c_i}
    - \frac{\big( b_i \Delta a + a_i \Delta b + \Delta a \Delta b \big) \Delta c}{c_i^2}
    + \frac{a_{i+1} b_{i+1}}{c_i^2 c_{i+1}} (\Delta c)^2 .
\end{equation*}
Take conditional expectations given $\mathcal{F}_i$. The first three terms on the right-hand side are linear in a single increment with $\mathcal{F}_i$-measurable coefficients, and hence vanish by the martingale property of $(a_i)$, $(b_i)$, $(c_i)$. Applying the triangle inequality and bounding the remaining terms in absolute value via Assumption (i), using $c_i, c_{i+1} \geq \varepsilon$, $\lvert a_{i+1} b_{i+1} \rvert \leq AB$, and $\lvert \Delta a \rvert \leq 2A$ (as $\lvert a_i \rvert, \lvert a_{i+1} \rvert \leq A$),
\begin{equation*}
    \big| \mathbb{E}\big[ Z_{i+1} - Z_i \,\big|\, \mathcal{F}_i \big] \big|
    \leq \frac{1}{\varepsilon}\, \mathbb{E}\big[ \lvert \Delta a \Delta b \rvert \,\big|\, \mathcal{F}_i \big]
    + \frac{B}{\varepsilon^2}\, \mathbb{E}\big[ \lvert \Delta a \Delta c \rvert \,\big|\, \mathcal{F}_i \big]
    + \frac{3A}{\varepsilon^2}\, \mathbb{E}\big[ \lvert \Delta b \Delta c \rvert \,\big|\, \mathcal{F}_i \big]
    + \frac{AB}{\varepsilon^3}\, \mathbb{E}\big[ (\Delta c)^2 \,\big|\, \mathcal{F}_i \big].
\end{equation*}
By the conditional Cauchy--Schwarz inequality, the inequality of arithmetic and geometric means, and Assumption (ii), $\mathbb{E}[ \lvert \Delta a \Delta b \rvert \mid \mathcal{F}_i ] \leq (\xi_i^a + \xi_i^b)/2$, and analogously for the other mixed terms. Collecting coefficients and bounding each by $K$ gives \eqref{eq:acid-ratio-bound}. The sequence $(\xi_i)_{i \geq 1}$ is non-negative and $\mathcal{F}$-adapted since $\xi_i^a$, $\xi_i^b$, $\xi_i^c$ are, and it is a.s.\ summable as a finite linear combination of a.s.\ summable sequences.
\end{proof}

Now, we use \Cref{lem:acid-ratio} to prove the main theorem.

\begin{proof}[Proof of \Cref{thm:iv_convergence}]
Let $\mathcal{F}_i = \sigma(Y_{1:i}, X_{1:i}, Z_{1:i})$ for $i \geq 1$ denote the filtration generated by the observed and imputed data. We show that the predictive sequence
\begin{align}
p_i(y(0) \mid CP) &= \frac{p_{NT, i} + p_{CP, i}}{p_{CP, i}} p_i(y \mid X = 0, Z = 0) - \frac{p_{NT, i}}{p_{CP, i}} p_i(y \mid X = 0, Z = 1) \nonumber \\
&= \frac{p_{NT, i}}{p_{CP, i}} \left[p_i(y \mid X = 0, Z = 0) - p_i(y \mid X = 0, Z = 1) \right] + p_i(y \mid X = 0, Z = 0) \label{eq:complier_sequence}
\end{align}
for $i \geq 1$ satisfies the conditions in \cite{battiston2025bayesianpredictiveinferencemartingales} and therefore has a well-defined limiting measure. The proof for $p(y(1) \mid CP)$ is analogous.

The numerator and denominator in the weights term of \Cref{eq:complier_sequence} are linear functionals of $p_i(x \mid z)$ and inherit the martingale property of the treatment update. Similarly, the difference of outcome densities remains a martingale. Thus, it is only left to show that the assumptions of \Cref{lem:acid-ratio} are satisfied for the left term.

Assumption (i) holds with $A = 1$ (the weights are probabilities), $B = \bar{p}$ (the uniform outcome-density bound from the proof of \Cref{thm:convergence}), and $c_i = p_{CP,i} \geq \varepsilon > 0$ by the monotonicity assumption (Assumption \ref{ass:iv_monotonicity}). For Assumption (ii), the Bayesian bootstrap increments are bounded by $1/i$ and thus the squared increments are summably bounded. For the copula update, the bound $|\Delta_i(y \mid x, z)| \leq \bar{\Delta}/i$ from the proof of \Cref{thm:convergence} gives a summable $\bar{\Delta}^2/i^2$ bound on the squared outcome-density increments.

Thus, \Cref{lem:acid-ratio} applies and the first term in \Cref{eq:complier_sequence} is an almost supermartingale with summable errors. Adding a martingale term preserves the property, so its random measure converges weakly almost surely by \citet[Theorem 1]{battiston2025bayesianpredictiveinferencemartingales}.

\end{proof}

\section{An extension to general instrumental variable models} \label{sec:iv_extension}

Consider the outcome $Y \in \mathcal{Y} \subseteq \mathbb{R}$, the treatment $X \in \mathcal{X} \subseteq \mathbb{R}$, and a vector of instruments $Z \in \mathcal{Z} \subseteq\mathbb{R}^p$ generated by the structural causal model
\begin{align}
    \begin{aligned}
        Y &= f(X) +\epsilon \\
        X &= g(Z) + \eta,
    \end{aligned}
\end{align}
where $\epsilon$ and $\eta$ are correlated error terms. Additional (exogenous) covariates can be added as arguments of both $f(\cdot)$ and $g(\cdot)$, in which case we have the nonlinear and heterogeneous model of \cite{spanbauer_flexible_2024}, which they analyse using  Bayesian additive regression trees (BART). Compared to \cite{holovchak_distributional_2025}, our proposed approach is slightly less general as we restrict ourselves to additive error terms, whereas they can consider arbitrary monotone functions. However, our method naturally quantifies epistemic uncertainty about the interventional distribution.

Under the assumption that $Z$ is exogenous, {i.e.}~that $Z \ind (\epsilon, \eta)$, we can decompose the error $\epsilon$ into a conditional expectation and an additive noise term $\epsilon = \E{\epsilon \mid \eta} + \nu$, where $\nu$ is independent of $\eta$ (and thus $X$). The observational model can then be written as
\begin{align*}
Y = f(X) + \E{\epsilon \mid \eta} + \nu.
\end{align*}
If the structural residuals follow a bivariate Gaussian distribution, the control function $\E{\epsilon \mid \eta}$ is a linear function of $\eta$, but we allow for more general dependence structures. This motivates the following strategy to estimate the interventional distribution: we estimate $g$ in the first-stage regression and obtain an estimate of $\eta$, which we use to update the observational distribution $p(y \mid x, \eta)$. The interventional distribution is then obtained by integrating over the marginal distribution of $\eta$
\begin{align*}
    p(y(x)) = \int p(y \mid x, \eta) p(\eta) \mathrm{d}\eta.
\end{align*}
Heuristically, this strategy is justified if the treatment assignment is ignorable conditional on an unobserved confounder and the residual $\eta$ provides a good approximation to that unobserved confounder. In addition, we require sufficient overlap so that the probability of receiving any level of treatment is positive for all values of the control function.

This strategy can be easily implemented within the martingale posterior framework. Given observed data $\{(y_i, x_i, z_i)\}_{i=1}^n$, we estimate $\hat{g}_n$, $\{\eta_i\}_{i=1}^n$, and $p_n(y \mid x, \eta)$. In each predictive sequence, we resample the instruments $Z$ and the treatment $X$ with the Bayesian bootstrap, update the structural function $\hat{g}_{i}$ to $\hat{g}_{i+1}$, estimate $\eta_{i+1}$, and then update $p_i(y \mid x, \eta)$ to $p_{i+1}(y \mid x, \eta)$. The interventional distribution is obtained by averaging the final conditional density over the empirical distribution of the recovered residuals:
\begin{align*}
p_N(y(x)) \approx \frac{1}{N} \sum_{i=1}^N p_N(y \mid x, \eta_i).
\end{align*}
If we assume that the first-stage function is linear, potentially using some basis expansion $\Phi: \mathcal{Z} \to \mathbb{R}^m$ such that $g(Z) = \Phi(Z)^\intercal \beta$, then the predictive update can be performed by recursively updating the least-squares estimate. 

\subsection{Simulated example}

We illustrate the general IV construction on simulated data generated from the structural model
\begin{align*}
  Y &= f(X) + \varepsilon \\ 
  X &= \pi Z + \eta
\end{align*}
with $f(x) = 2\tanh(1.5x)$ and $Z \sim \mathrm{N}(0, 1)$. The residuals are simulated as $\eta \sim \mathrm{N}(0, 1)$ and $\varepsilon = c_1 \eta + c_2(\eta^2 - 1) + \nu$, where $\nu \sim \mathrm{N}(0, \sigma_\nu^2)$ is independent of $(Z, \eta)$. We set $\pi = 1$, $c_1 = 0.7$, $c_2 = 0.35$, and fix $\sigma_\nu^2 = 1 - c_1^2 - 2c_2^2 = 0.265$ so that $\operatorname{Var}(\varepsilon) = 1$. This results in a right-skewed interventional density $p(y(x))$, whereas the observational distribution is more symmetric. The marginal densities are not available in closed form, but can be easily evaluated numerically to obtain the ground truth.

We generate a single sample of $n = 500$ observations and target the interventional distribution at the treatment level $x = -1$. For comparison, we also report the martingale posterior for the observational density $p(y \mid X = x)$, which ignores the instrument. Figure~\ref{fig:iv_sim_illustration} shows that our proposed method recovers both densities well. In particular, it captures the additional skewness of the interventional distribution.

\begin{figure}
    \centering
    \includegraphics[width=0.8\linewidth]{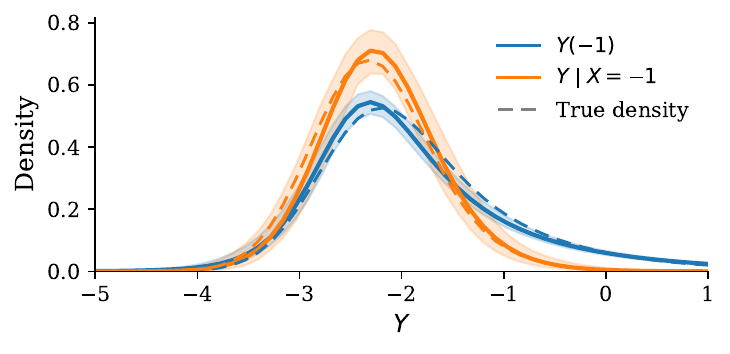}
    \caption{\textbf{Simulated data.} Martingale posterior estimates (with $95\%$ credible intervals) of the interventional densities $p(y(x))$ and the observational density $p(y \mid x)$ at $x = -1$. The results are based on a single sample of size $n = 500$ and $B = 100$ predictive sequences of $2{,}000$ forward samples each.}
    \label{fig:iv_sim_illustration}
\end{figure}


\end{document}